\documentclass[11pt]{article}
\usepackage[margin=1in]{geometry}
\usepackage{amsmath, amssymb, amsthm}
\usepackage{physics}
\usepackage{tikz}
\usepackage{quantikz} 

\usepackage{graphicx}
\usepackage{hyperref}
\usepackage{cleveref}
\hypersetup{
    colorlinks=true,
    linkcolor=blue,
    citecolor=blue,
    urlcolor=blue
}

\usepackage{booktabs} 

\usepackage{natbib}
\usepackage{authblk} 

\newcommand{\aqa}{$\langle aQa^L \rangle$ Applied Quantum Algorithms, Leiden University, The Netherlands}
\newcommand{\liacs}{LIACS, Leiden University, Niels Bohrweg 1, 2333 CA, Leiden, The Netherlands}
\newcommand{\bmw}{BMW Group, 80788 München, Germany}

\newtheorem{lemma}{Lemma}
\newtheorem{theorem}{Theorem}

\newtheorem{definition}{Definition}
\newtheorem{remark}{Remark}
\newtheorem{example}{Example}

\title{Universality and kernel-adaptive training for classically trained, quantum-deployed generative models}

\author[1,2,3]{Andrii Kurkin}
\author[1,2,3]{Kevin Shen}
\author[3]{Susanne Pielawa}
\author[1,2]{Hao Wang}
\author[1,2]{Vedran Dunjko}

\affil[1]{\aqa}
\affil[2]{\liacs}
\affil[3]{\bmw}

\date{} 

\begin{document}
\maketitle

\begin{abstract}
    The instantaneous quantum polynomial (IQP) quantum circuit Born machine (QCBM) has been proposed as a promising quantum generative model over bitstrings. Recent works have shown that the training of IQP-QCBM is classically tractable w.r.t. the so-called Gaussian kernel maximum mean discrepancy (MMD) loss function, while maintaining the potential of a quantum advantage for sampling itself. Nonetheless, the model has a number of aspects where improvements would be important for more general utility: (1) the basic model is known to be \textit{not} universal - i.e. it is not capable of representing arbitrary distributions, and it was not known whether it is possible to achieve universality by adding \textit{hidden} (ancillary) qubits; (2) a fixed Gaussian kernel used in the MMD loss can cause training issues, e.g., vanishing gradients. In this paper, we resolve the first question and make decisive strides on the second. We prove that for an $n$-qubit IQP generator, adding $n + 1$ \textit{hidden} qubits makes the model universal. For the latter, we propose a \textit{kernel-adaptive} training method, where the kernel is adversarially trained. We show that in the kernel-adaptive method, the convergence of the MMD value implies weak convergence in distribution of the generator. We also analytically analyze the limitations of the MMD-based training method. Finally, we verify the performance benefits on the dataset crafted to spotlight improvements by the suggested method. The results show that kernel-adaptive training outperforms a fixed Gaussian kernel in total variation distance, and the gap increases with the dataset dimensionality. These modifications and analyses shed light on the limits and potential of these new quantum generative methods, which could offer the first truly scalable insights in the comparative capacities of classical versus quantum models, even without access to scalable quantum computers.
\end{abstract}
    
\section{Introduction}

Quantum computing holds great promise for advancing machine learning, yet achieving a practical quantum advantage remains a significant challenge. While theoretical results have shown separations in learning capabilities—often grounded in cryptographic hardness assumptions~\cite{liu2021rigorous, gyurik2023exponential}, such proofs are not expected to be necessarily relevant in practical settings. In real-world applications, a model's value is ultimately measured by its empirical performance. However, current quantum hardware limitations make it practically impossible to evaluate a model, not to mention training, beyond a small qubit count. This constraint prevents any meaningful assessment of their performance on tasks of a realistic scale.

Surprisingly, a specific quantum generative model can be efficiently trained and, to an extent, evaluated at scale. The model is a type of Quantum Circuit Born Machine (QCBM) \cite{benedetti2019generative, liu2018differentiable}, with the circuit specifically chosen as a parameterized instance of Instantaneous Quantum Polynomial (IQP) circuit \cite{Bremner_2010, nakata2014diagonal, recio2025train}. The model prepares an $n$-qubit state from a parameterized IQP circuit, and measuring each qubit in the computational basis will generate a bitstring of length $n$. This corresponds to sampling from the underlying circuit distribution. We call this model IQP-QCBM. This model has been demonstrated to be efficiently trainable at scale on classical hardware under a specific training objective—namely, the Maximum Mean Discrepancy (MMD) with a Gaussian kernel. This was shown on high-dimensional datasets with models of up to 1\,000 qubits~\cite{recio2025train, recio2025iqpopt}, achieving performance comparable to or better than classical counterparts. Moreover, this model exhibits provable classical hardness when used as a sampler~\cite{Bremner_2010, Bremner_2016, marshall2024improved}, leaving room for quantum advantage. However, despite these advances, several challenges remained. 

One of the most fundamental properties of a machine learning model family is universality. In generative modeling it is the ability to approximate arbitrary distributions given sufficient model resources. Universality has been established for several important families, including Restricted Boltzmann Machines \cite{LeRouxBengio2008_RBM_UA,MontufarAy2011_RBM_DBN_refinements}, GAN-style generators \cite{LiuBousquetChaudhuri2017_GAN_convergence,LuLu2020_UA_distributions}, and diffusion/score-based models\cite{SongEtAl2021_ScoreSDE,LeeLuTan2023_SGM_convergence}. In quantum generative modeling, parameterized quantum-circuit generators have been shown to be universal for continuous multivariate distributions, with explicit resource trade-offs \cite{BartheEtAl2025_PQC_universal_generative}; related work characterizes the expressive power of such circuits and quantum GANs \cite{DuEtAl2020_PQC_expressive_power,DallaireDemersKilloran2018_QGAN}. However, the basic $n$ qubits IQP-QCBM architecture lacks the expressive power to model arbitrary probability distributions. Specifically, although the model has $2^n - 1$ parameters available to tune, the model as defined in \cite{recio2025train} is \emph{not} universal over the space of probability distributions on $n$-bit strings, the same authors conjectured that universality might be achieved by augmenting the system with additional hidden (ancillary) qubits and considering marginal distributions over the visible ones. In this paper, we resolve this conjecture by providing two constructive proofs that IQP-QCBMs with hidden qubits can, in fact, represent arbitrary distributions over the visible qubits. Our first concept proof construction demonstrates asymptotic universality, where the total variation distance between the model and target distribution decreases exponentially with the number of hidden qubits and requires just a restricted set of relative phases. Our second construction achieves \emph{exact} universality by simply doubling the number of qubits and uses the full $[0, 2\pi)$ phase range.

Another limitation of the original approaches involves the lack of flexibility with respect to the kernel that can be used in MMD-based training. The training method proposed in~\cite{recio2025train} relies on a representation of the Gaussian kernel as a mixture of Pauli-$Z$ observables, which enables efficient classical estimation. However, MMD with a Gaussian kernel is not always sensitive to distinguishing between distinct distributions, which may yield vanishing gradients, making optimization difficult. To address this, we generalize the MMD objective to arbitrary kernels and show that it remains efficiently estimable using classical algorithms. Building on this, we suggest leveraging a kernel-adaptive adversarial training procedure for IQP-QCBMs. While the idea of trainable kernels was studied for classical generative models before~\cite{li2017mmd, li2019implicit, arbel2018gradient, mroueh2021convergence}, we are the first to propose adaptive kernels for the IQP-QCBM model by parameterizing the spectral measure of the kernel. Our method tailors the choice of kernel, or its spectral measure, to be precise, to the specific learning task, which might enhance the convergence. We support this approach through both theoretical analysis and empirical validation.

The structure of the paper is as follows. In~\cref{sec: prelims}, we review the background on IQP-QCBMs. \cref{seq: universality} presents our universality results for models with hidden qubits. In~\cref{sec: kernel adaptive training}, we introduce kernel-adaptive training and provide theoretical justifications for it. Also, we provide limitations of MMD-based training in \cref{seq: MMD limitations}. The \cref{seq: numerics} contains our numerical experiments on the benchmark dataset. We conclude with a discussion in~\cref{seq: discussion}.

\section{Preliminaries}\label{sec: prelims}
We focus on a specific class of quantum generative models, Instantaneous Quantum Polynomial (IQP) Quantum Circuit Born Machine (QCBM), which is of great interest: it demonstrates a provable separation between quantum and classical sampling algorithms~\cite{Bremner_2010,Bremner_2016}. In this paper, we shall refer to it as IQP-QCBM. We recap two interesting and important properties of it: (1) non-universality (\cref{subsec:non-universality}) and (2) with the maximum mean discrepancy (MMD) metric, it is classically trainable (\cref{subsec:non-universality}).

\subsection{IQP-QCBM model}\label{subsec: model} 
\begin{definition}[Quantum Circuit Born Machine (QCBM)]
A QCBM is a parameterized quantum generative model specified by: (1) A parameterized quantum circuit $U(\theta)$ acting on $n$ qubits, preparing the state $U(\theta) \ket{0}^{\otimes n}$; (2) Measurement in the computational basis yields bitstrings $x$ with probability: $q_{\theta}(x) = |\langle x | \psi(\theta) \rangle|^2$ which constitutes the output distribution of the model.
\end{definition}

Here, we use the following parameterization of $U(\theta)$.
\begin{definition}[Parametrized instantaneous quantum polynomial (IQP) circuit]
A parametrized IQP circuit on $n$ qubits is a quantum circuit of the form $U(\theta) = H^{\otimes n} D(\theta) H^{\otimes n}$,
where $H$ denotes the Hadamard gate, and $D(\theta) = \prod_j e^{i\theta_j Z_{g_j}}$ with $\theta_j\in[0, 2\pi)$, where $Z_{g_j}$ is a tensor product of Pauli-$Z$ operators acting on a subset of qubits specified by the nonzero entries of $g_j \in \{0,1\}^n$.
\end{definition}

Theoretically, there exists an output distribution of IQP-QCBM that is computationally intractable for classical algorithms~\cite{Bremner_2010}, under standard complexity-theoretic assumptions. Furthermore, this quantum-classical separation extends to approximate sampling as well~\cite{Bremner_2016, marshall2024improved}.
Apart from the sampling hardness, it is necessary to investigate IQP-QCBM universality as a generative model.
\subsection{Non-universality}\label{subsec:non-universality}
A key property of generative models is universality: informally, with sufficient model capacity they can approximate any probability distribution. We give a rigorous definition below. 

\begin{definition}[Universality of generative models]
A family of generative models $Q$ is \textit{universal} if for any target distribution $p(x), x\in \mathcal{X}$ ($\mathcal{X}$ is a topological space) and any precision $\varepsilon > 0$, there exist model $q_{\theta}\in Q$ and a parameter setting $\theta$ such that $d(p, q_{\theta}) \leq \varepsilon$, where $d(\cdot,\cdot)$ is a metric between probability distributions, e.g., total variation, Wasserstein. We say universality is \emph{exact} if the condition holds with $\varepsilon = 0$ i.e. there exists $q_{\theta} \in Q$ such that $d(p, q_{\theta}) = 0$.

\end{definition}

IQP-QCBMs are not universal, as shown in~\cite{recio2025train}: $n$-qubit models cannot represent any distribution over the Boolean hypercube $\{0,1\}^n$. For example, the distribution $p = \left(\frac{1}{3}, \frac{1}{3}, \frac{1}{3}, 0\right)$ over $\{0,1\}^2$ cannot be represented by $2$-qubit model. More generally, we notice that \textit{no} distribution over $\{0,1\}^2$ with support size $3$ is expressible by a $2$-qubit IQP-QCBM, see derivations in~\cref{proofs: Hidden v.s. visible only: toy example}. This non-universality constrains the model's representational capacity, motivating the need to look at a universal extension of IQP-QCBM introduced in \cref{seq: universality}.

\subsection{Classical training}\label{subsec:classical-training}
We recap MMD's definition on IQP-QCBM and then show it can be estimated efficiently.
\begin{definition}[Maximum Mean Discrepancy (MMD)]
Given two distributions $p$ and $q$ over space $\mathcal{X}$, and a kernel $k: \mathcal{X} \times \mathcal{X} \to \mathbb{R}$ which induces a reproducing kernel Hilbert space (RKHS) $\mathcal{H}$, the MMD metric is:
\begin{equation}\label{MMD_def}
    \operatorname{MMD}(p, q) = \sup_{\substack{f \in \mathcal{H},  \|f\|_{\mathcal{H}} \leq 1}} \left( \mathbb{E}_{x \sim p}[f(x)] - \mathbb{E}_{y \sim q}[f(y)] \right),
\end{equation}
where RKHS norm $\|f\|_{\mathcal{H}}$ is defined with the kernel $k$~\citep{muandet2017kernel}.
\end{definition}
With the Gaussian kernel, i.e., $k(b, b')=\exp(-\sum_{i=1}^n|b_i - b'_i| / 2\sigma^2), b, b'\in\{0,1\}^n$, the MMD metric of the IQP-QCBM model admits the following form~\citep{rudolph2024trainability}:
\begin{equation}\label{eq: Gaussian MMD}
    \operatorname{MMD}^2(p, q_\theta) = \mathbb{E}_{\alpha \sim G_\sigma} \left[ \left( \langle Z_\alpha \rangle_p - \langle Z_\alpha \rangle_{q_\theta} \right)^2 \right],
\end{equation}
where $Z_\alpha\in\mathbb{C}^{2^n\times 2^n}$ is a Pauli-$Z$ operator acting non-trivially on the qubits indexed by the nonzero entries of bitstring $\alpha \in \{0,1\}^n$, and $\langle Z_\alpha \rangle_{q_\theta} = \mathbb{E}_{b \sim q_\theta} \langle b | Z_\alpha | b \rangle$ ($q_\theta$ is the output distribution of IQP-QCBM. The distribution $G_\sigma$ is:
\begin{equation}\label{eq: Gaussian spectral_measure}
    G_\sigma(\alpha) = (1 - p_\sigma)^{n - |\alpha|} p_\sigma^{|\alpha|}, \quad \alpha \in \{0, 1\}^n,
\end{equation}
with $p_\sigma = \frac{1}{2}\left(1 - \exp(-1/2\sigma)\right)$ and $|\alpha|$ is the Hamming weight of $\alpha$.

Crucially, the expectation $\langle Z_\alpha \rangle_{q_\theta}$ can be estimated efficiently with an classical algorithm~\citep{nest2009simulating,recio2025train}.

\begin{lemma}\label{lemma: efficient Pauli-Z estimation IQP}
Given a parameterised IQP circuit $q_\theta$, an expectation value $\langle Z_\alpha \rangle_{q_\theta}$, and an error $\varepsilon\in\mathcal{O}(\mathrm{poly}(n^{-1}))$, there exists a classical algorithm that requires $\mathrm{poly}(n)$ time, and samples a random variable with standard deviation less than $\varepsilon$ that is an unbiased estimator of $\langle Z_\alpha \rangle_{q_\theta}$.
\end{lemma}

Despite that, we can train the model classically w.r.t. Gaussian kernel MMD, we point out that this kernel choice might bring a limitation: the corresponding distribution $G_\sigma$ has an exponential tail w.r.t. the Hamming weight (see~\cref{eq: Gaussian spectral_measure}), making it hard to distinguish two distributions when $\langle Z_\alpha \rangle_p - \langle Z_\alpha \rangle_{q_\theta}$ is concentrated on the large or small Hamming weights. 
As we will prove in~\cref{sec: kernel adaptive training}, the distribution $G_\sigma$ is actually the \emph{spectral measure} of the kernel $k$ (via Bochner's theorem), allowing for using non-Gaussian kernels in the training.

\section{Universality with hidden qubits} \label{seq: universality}
In this section, we define the extended version of the original IQP-OCBM model by incorporating hidden qubits. We present two key results regarding their universality: one approximate-asymptotic and the stronger one exact. The last one, based on construction, is quite frugal, requiring only doubling the qubit number.

\begin{definition}[IQP-QCBM with hidden qubits]
A parameterized IQP circuit with hidden qubits is an IQP circuit in which a designated subset of qubits (called \emph{hidden qubits}) is traced out prior to measurement.
Let the full system have $m + n$ qubits, where the first $m$ qubits are designated as hidden.
The output distribution is obtained by taking the partial trace over the hidden qubits:
\begin{equation}
    q(x) = \operatorname{Tr}\left({\operatorname{Tr}_{\text{hidden}}\left(\rho \right)\ket{x}\bra{x}}\right),
\end{equation}
where $x \in \{0,1\}^n$ is a computational basis state, and $\rho$ is the quantum state produced by IQP.
\end{definition}

\begin{figure}[t]
  \centering
  \begin{minipage}{0.42\textwidth}
    \centering
  \resizebox{\linewidth}{!}{%
      \begin{tikzpicture}[baseline=(current bounding box.north)]
        \node (circuit) {
          \begin{quantikz}[row sep=0.2cm, column sep=0.2cm]
            \lstick[wires=3]{$\ket{0}^{\otimes m}$}
                  & \gate{H} & \gate[wires=6]{D(\theta)} & \qw \\
            \setwiretype{n}
                  & \push{\vdots} &                  & \push{\vdots} & \\
            \setwiretype{q}
                  & \gate{H} &                           & \qw \\
            \lstick[wires=3]{$\ket{0}^{\otimes n}$}
                  & \gate{H} &                          & \gate{H} & \meter{} \\
            \setwiretype{n}
                  & \push{\vdots} &                  & \push{\vdots} & \\
            \setwiretype{q}
                  & \gate{H} &                          & \gate{H} & \meter{}
          \end{quantikz}
        };
        \fill[gray, opacity=0.2] (-2.7cm,0cm) rectangle (2.8cm,2.4cm)
          node[xshift=-.6cm, yshift=-1.3cm, opacity=0.8]{hidden};
      \end{tikzpicture}%
    }
    \caption{Schematic diagram of IQP-QCBM with hidden qubits.}
    \label{fig:fig2}
  \end{minipage}
\hfill
    \begin{minipage}{0.5\textwidth}
    \centering
    \includegraphics[width=\linewidth]{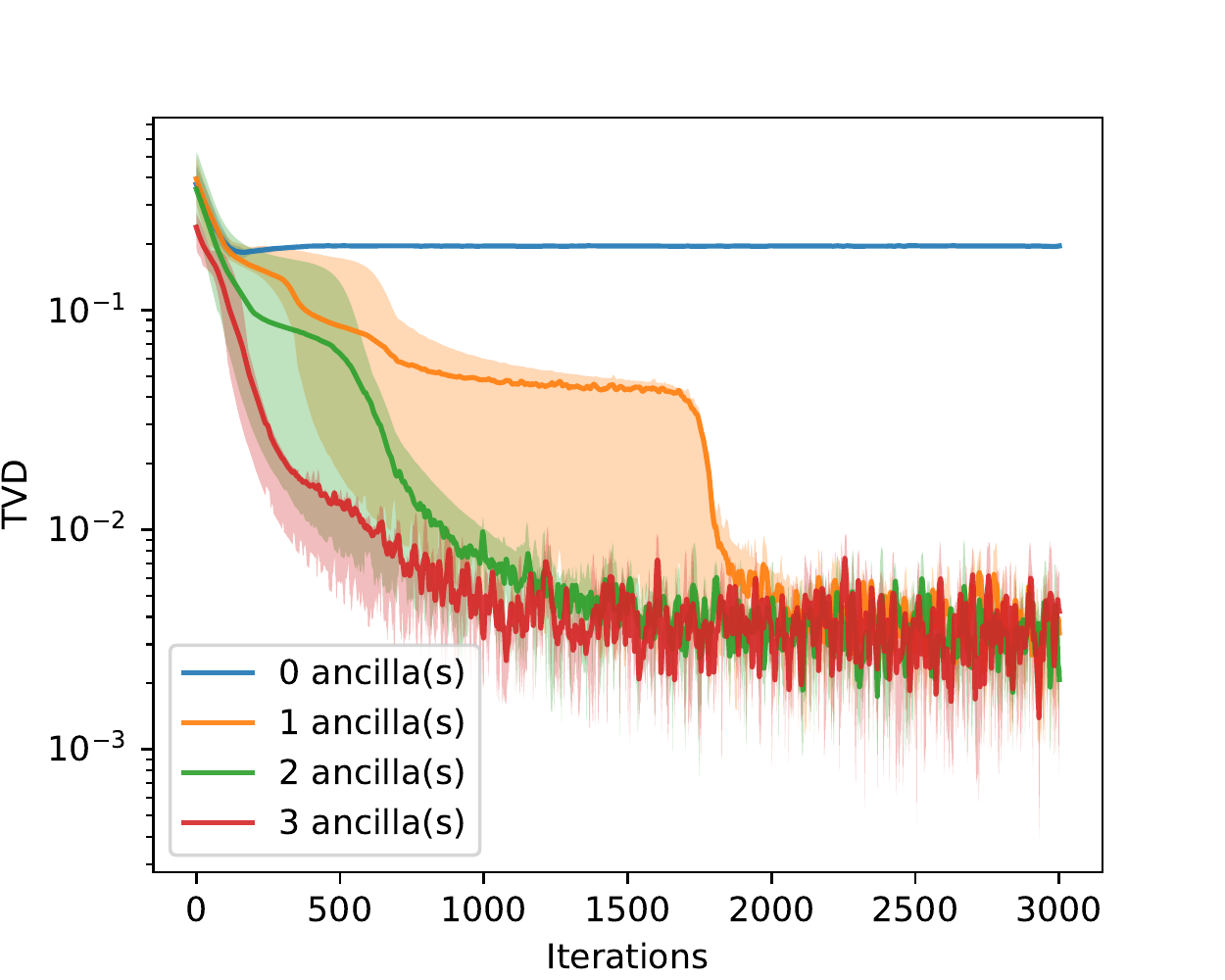}
    \caption{Example to show the expressivity enhancement via adding hidden/ancilla qubits for target distribution vector $p = \left(\frac{1}{3}, \frac{1}{3}, \frac{1}{3}, 0\right)$ over $\{0,1\}^2$. With just one hidden qubit, the total variation distance (TVD) between the generator and the target is greatly reduced over the training iterations.}
    \label{fig: 2-qubit iqp tvd}
  \end{minipage}
\end{figure}

This construction draws inspiration from the structure of Boltzmann machines, where hidden neurons are marginalized out to increase expressivity. Similarly, introducing hidden qubits in IQP circuits allows for universal modeling of probability distributions. Our first proof is asymptotic.

\begin{lemma}[Asymptotic approximate universality with $\pi$-phases]\label{lemma: asymptotic approximate universality}
For any probability distribution $p$ over $\{0,1\}^n$, there exists a parameter setting $\theta=\{\theta_{j,k}\}_{j\in\{0,1\}^m,\,k\in\{0,1\}^n}$ for the IQP circuit acting on $m+n$ qubits such that the total variation distance between the marginal distribution $q_\theta$ on the final $n$ qubits and the target distribution $p$ satisfies $\mathrm{TVD}(p,q_\theta)\in\mathcal{O} \left(\frac{1}{2^{m-n}}\right)$.
\end{lemma}

See the proof in~\cref{proof: asymptotic approximate universality}. The basic idea is to show that we can use ancillary registers to provide arbitrary counts of bitstrings in the output register, which, after normalization, achieves arbitrary frequencies of bitstrings, which then approximate the target distribution. This statement is not the most efficient possible, but utilizes only $0$ and $\pi$ relative phases. By using all possible phases, a stronger result can be obtained.

\begin{theorem}[Exact universality]\label{thm: exact universality}
For any target probability distribution $p$ over $\{0,1\}^n$, there exists an IQP circuit acting on $n$ visible qubits and $m = n+1$ hidden qubits that produces $p$ exactly when we measure the final $n$ qubits.
\end{theorem}

\begin{proof}[Proof sketch]
First, we show that tracing out the hidden $m$ qubits gives a reduced density matrix:
\begin{equation}
    \rho_2 \coloneqq\operatorname{Tr}_{\text{hidden}}(\rho)= \frac{1}{2^m} \sum_{k \in \{0, 1\}^m} |\psi_k\rangle \langle \psi_k|,
\end{equation}
where $|\psi_k\rangle = \frac{1}{\sqrt{2^n}} \sum_{y \in \{0, 1\}^n} e^{i\theta_{k,y}} |y\rangle$.
We show that it is possible to encode any 2-sparse distribution (i.e. supported on at most two outcomes) over $\{+,-\}^n$ into $|\psi_k\rangle$ by tuning the parameters.
The probability of sampling string $b \in \{0,1\}^n$ is $|\langle \tilde{b}| \psi_k\rangle|^2$ with $|\tilde{b}\rangle = H^{\otimes n} |b\rangle$.
Next, we prove that any probability distribution over $\{0,1\}^n$ can be decomposed as a uniform mixture of at most $2^{n+1}$ 2-sparse distributions. Therefore, by choosing $m = n+1$, we can represent any target distribution $p$ exactly using an IQP circuit with hidden qubits. See the~\cref{proof: exact universality} for the complete proof.  
\end{proof}

As illustrated in \cref{fig: 2-qubit iqp tvd}, for the 2-bit target distribution $p=(\tfrac{1}{3},\tfrac{1}{3},\tfrac{1}{3},0)$ over $\{0,1\}^2$, adding just one hidden qubit (orange curve) enables an IQP–QCBM to achieve a substantially lower total variation distance than a baseline IQP model without hidden qubits. Moreover, we show that one hidden qubit is theoretically sufficient to represent \emph{any} distribution over two bits, see \cref{proofs: Hidden v.s. visible only: toy example}. However, for $n>2$ it still remains open whether the requirement of $n{+}1$ hidden qubits for universality is optimal. We also highlight that the universal construction suggested above is purely theoretical, entails exponentially many parameters, and enlarges the register from $n$ to $2n{+}1$ qubits, which immediately kills trainability — so it is not a practical prescription. In practice, far fewer hidden qubits typically suffice to boost expressivity and capture many distributions of interest.

\begin{remark}\label{remark: efficient Pauli-Z estimation IQP with hidden qubits}
    Similar to \cref{lemma: efficient Pauli-Z estimation IQP}, we can efficiently estimate the expectation value of Pauli-$Z$ words w.r.t. the distribution produced by IQP-QCBM with \textit{hidden} qubits: given a parameterized IQP-QCBM model on $m$ hidden and $n$ visible qubits, which produces a distribution $q_\theta$ over $\{0,1\}^n$, an expectation value $\langle Z_\alpha \rangle_{q_\theta}$. Given an error $\varepsilon\in\mathcal{O}(\mathrm{poly}(n^{-1}, m^{-1}))$, there exists a classical algorithm that requires $\mathrm{poly}(n, m)$ time, and samples a random variable with standard deviation less than $\varepsilon$ that is an unbiased estimator of $\langle Z_\alpha \rangle_{q_\theta}$. We include the precise argument in the~\cref{proof: efficient Pauli-Z estimation IQP with hidden qubits}.
\end{remark}

\section{Kernel-adaptive training}\label{sec: kernel adaptive training}
In this section, we introduce a kernel-adaptive adversarial training procedure for IQP-QCBM, designed to better align the model with the specific learning task of interest. First, we begin with an example to show that MMD with a fixed Gaussian kernel fails: there exist two distributions where the Gaussian MMD value decays exponentially while the total variation distance is constant.  We note that the MMD distance is estimated via sampling, so exponential precision is not achievable in polynomial time. Second, we derive a generalized MMD expression applicable to arbitrary kernels, which can be efficiently estimated with classical algorithms. Third, we present the adversarial training method, which adaptively learns the kernel on the fly. 

\subsection{Limitations of the Gaussian MMD loss}
We show that MMD with the Gaussian kernel might be ineffective in distinguishing probability distributions, resulting in a near-zero MMD value and vanishing gradients of the generator.

\begin{lemma}[Kernel choice matters]\label{lemma: kernel choice matters}
    For every $n \in \mathbb{N}$, there exist a target distribution $p$ on $\{0,1\}^n$, an IQP circuit on $n$ qubits with parameters $\theta$, and a setting $\theta'$ yielding distribution $q_{\theta'}$, such that: (1) $\operatorname{TVD}(p,q_{\theta'}) \in \Omega(1)$; (2) For a Gaussian kernel $k_\sigma$: $\operatorname{MMD}^2_{k_\sigma}(p,q_{\theta'}),  \big\|\nabla_\theta \operatorname{MMD}^2_{k_\sigma}(p,q_{\theta'})\big\| \in \mathcal{O}(2^{-n})$; (3) For a characteristic kernel $\kappa$, there exists a constant $C > 0$ such that for all $n$: $\operatorname{MMD}^2_{\kappa}(p,q_{\theta'})$, $\big\|\nabla_\theta \operatorname{MMD}^2_{\kappa}(p,q_{\theta'})\big\| \geq C$.
\end{lemma}
See \cref{proof: kernel choice matters} for the proof. The proof's key insight is: there exist distinct distributions $p$ and $q_{\theta'}$ such that the difference of their characteristic functions, $|\langle Z_\alpha \rangle_p - \langle Z_\alpha \rangle_{q_{\theta'}}|$, is supported on a few bitstring $\alpha$ of very small or very large Hamming weights. Now, consider the Gaussian kernel whose spectral measure $G_\sigma(\alpha)$ (\cref{eq: Gaussian spectral_measure}) decays exponentially with $n$ on both tails. In this case, $|\langle Z_\alpha \rangle_p - \langle Z_\alpha \rangle_{q_{\theta'}}|$ is supported on the tails of $G_\sigma(\alpha)$, which exponentially decays. In contrast, if we are allowed to deviate from the Gaussian, then we can choose a kernel $\kappa$ that concentrates on large or small Hamming weights, which prevents the exponential decay of the MMD value (w.r.t.~$n$). The same argument applies to the gradient of MMD easily by the chain rule. Practically, one can choose the kernel $\kappa$ adaptively to allocate more spectral mass to the bitstring where the generator $\langle Z_\alpha \rangle_{q_{\theta'}}$ differs the most from the data $\langle Z_\alpha \rangle_p$.  

\subsection{MMD with adaptive kernel} 
We generalize MMD's expression beyond the Gaussian kernel based on the spectral representation of MMD, as a direct consequence of Bochner's theorem \cite{bochner1933monotone}.
\begin{lemma}[Spectral representation of MMD]\label{lemma: spectral_MMD}
Let $X$ be a locally compact Abelian group and $p$, $q$ be probability distributions on $X$. 
Given a bounded, stationary kernel $k \colon X \times X \to \mathbb{R}$, and its Fourier transformation, $G$, which is a non-negative measure on the dual group $\widehat{X}$ (called the spectral measure of $k$)
Then the squared Maximum Mean Discrepancy admits the spectral representation:
\begin{equation}\label{lemma: MMD spectral representation}
    \mathrm{MMD}^2_k(p, q) = \mathbb{E}_{\alpha \sim G}\left[ \left| \phi_p(\alpha) - \phi_q(\alpha) \right|^2 \right],
\end{equation}
where $\phi_p(\alpha)=\mathbb{E}_{x\sim p } [\exp(-i\alpha \cdot x)]$ and $\phi_q(\alpha)=\mathbb{E}_{x\sim p } [\exp(-i\alpha \cdot x)]$ are the characteristic functions of $p$ and $q$ respectively. See~\cite{muandet2017kernel} for a proof.
\end{lemma}
We notice that expectation values of Pauli-$Z$ both for target and IQP-QCBM produced distributions are characteristic functions of corresponding distributions: $\langle Z_\alpha \rangle_p = \phi_p(\alpha)$ and $\langle Z_\alpha \rangle_{q_\theta} = \phi_{q_\theta}(\alpha$.
This combined with \cref{lemma: MMD spectral representation}, gives us:
\begin{theorem}[Generalized kernel MMD]\label{thm:generalized}
    For any stationary and bounded kernel $k$ over $\{0,1\}^n$, let $G$ be the Fourier transform of $k$. The $\mathrm{MMD}$ loss of IQP-QCBM can be expressed as:
    \begin{equation}\label{eq: MMD general kernel}
        \operatorname{MMD}_G^2(p, q_\theta) = \mathbb{E}_{\alpha \sim G} \left(\langle Z_\alpha \rangle_p - \langle Z_\alpha \rangle_{q_\theta}\right)^2
    \end{equation}
\end{theorem}
The above formulation generalizes the results in \cite{recio2025train}:
\begin{example}
    The \cref{lemma: MMD spectral representation} suggests that $G$ and $k$ are related via the Fourier transform. If we apply it to the Gaussian kernel, we obtain the following spectral mass:
    \begin{equation}
        G(\alpha) = \frac{1}{2^n} \sum_{b \in \{0,1\}^n} (-1)^{b \cdot \alpha} \exp(-\|b\|^2/2\sigma^2) = (1 - p_\sigma)^{n - |\alpha|} p_\sigma^{|\alpha|}
    \end{equation}
    which exactly coincide with \cref{eq: Gaussian spectral_measure} with the same $p_\sigma$.
\end{example}

\begin{remark}[Efficiency of generalized MMD estimation]
    For the IQP-QCBM model, it is possible to construct an unbiased estimator of $\mathrm{MMD}^2$ and its gradients w.r.t. $\theta$ using an efficient classical algorithm. This is because (1) the expectation value $\expval{Z_{\alpha}}_{q_\theta}$  can be efficiently estimated~\cref{lemma: efficient Pauli-Z estimation IQP}, and (2) $\mathrm{MMD}^2$ in Eq.~\eqref{eq: MMD general kernel} is a probabilistic mixture over these expectation values, which can be estimated efficiently by sampling bitstrings $\alpha\in\{0,1\}^n$ from the measure $G$.
\end{remark}

\subsection{Adversarial training with adaptive kernel}
To leverage adaptive, task-dependent kernels in MMD based on generalized representations, we require kernels to be stationary and bounded. However, meaningful training necessitates an additional constraint---the kernel must be \textit{characteristic}. A kernel $k$ over topological space $\mathcal{X}$ is characteristic if for any probability measures $\mu$ and $\nu$ in $\mathcal{X}$, the map $\mu \mapsto \int_{\mathcal{X}} k(x, \cdot) d\mu(x)$ is injective. This ensures distinct distributions yield unique kernel mean embeddings.

Being characteristic is essential for $\mathrm{MMD}$ loss, as it guarantees $\mathrm{MMD}^2(p, q) = 0$ \textit{if and only if} $p = q$~\citep{muandet2017kernel}. For binary spaces $\{0,1\}^n$, a kernel $k(t)$ is characteristic if and only if its spectral measure $G(\alpha)$ has full support ($\operatorname{supp}(G) = \{0,1\}^n$), meaning $G(\alpha) > 0$ for all $\alpha \in \{0,1\}^n$, see \cite{fukumizu2008characteristic}.
Instead of parameterizing the kernel, we model its spectral measure $G_\gamma$ using a critic network with parameters $\gamma$. This leads to an adversarial training scheme for IQP-QCBMs, analogous to generative adversarial networks~\cite{goodfellow2020generative}. We define the loss:
\begin{align}\label{eq: max MMD loss}
    \mathcal{L}(p, q_\theta)  &\coloneqq \max_{\gamma} \operatorname{MMD}_{G_\gamma}^2(p, q_\theta) 
\end{align}
with the constrained optimization:
\begin{equation}
    \begin{aligned}
        \min_{\theta} \,\mathcal{L}(\theta)\; \text{ s.t. } \sum_{\alpha \in \{0,1\}^n}G_\gamma(\alpha) = 1 \text{ and }G_\gamma(\alpha) > 0 \quad \forall \alpha \nonumber \\
    \end{aligned}
\end{equation}

In this min-max optimization setup, the critic $G_\gamma$ learns to identify the bistrings at which $\expval{Z_\alpha}_p$ and $\expval{Z_\alpha}_{q_\theta}$ differ most by adjusting $G_\gamma(\alpha)$. Unlike the fixed-kernel training procedure proposed in \cite{recio2025train}, in our approach, we assign higher importance to the $\alpha$ values that are most relevant to the task of interest. This adaptive weighting intuitively suggests improved performance by focusing learning capacity on the most informative spectral features.

For adaptive kernel loss~\cref{eq: max MMD loss}, it has been shown that the convergence in the loss value is equivalent to the convergence in distribution of probability measures for classical learning algorithms in real spaces~\citep{simon2023metrizing}. Here, we specialize this argument for IQP-QCBM.
\begin{lemma}[Consistency with weak convergence]\label{lemma: consistency with weak convergence}
Let $\{q_t\}_{t=1}^\infty$ be a sequence of probability distributions on $\{0,1\}^n$, and let $p$ be a fixed distribution on $\{0,1\}^n$. Let $G_\gamma: \{0,1\}^n \to (0,1)$ be any parametrized spectral measure with full support. Then the following equivalence holds:
\begin{equation}
\lim_{t \to \infty} \mathcal{L}(q_t, p) = 0 \iff q_t \overset{\text{d}}{\longrightarrow} p.
\end{equation}
where $\mathcal{L}$ is specified by \cref{eq: max MMD loss}.
\end{lemma}
See the proof in \cref{proof: consistency with weak convergence}. The lemma establishes a precise connection between the minimization of the trainable kernel MMD loss and the weak convergence of probability distributions: it tells us that if we imagine an idealized training process with an infinite number of steps, where at each step the model distribution $q_t$ is adjusted to reduce the MMD loss $\mathcal{L}(q_t, p)$ towards zero, then this iterative minimization guarantees that $q_t$ will eventually converge to $p$ in distribution. This means the MMD loss with adaptive kernel is not just a heuristic—it is \textit{consistent} with weak convergence.

\section{Limitations of the MMD loss} \label{seq: MMD limitations}
We now turn to the limitations of MMD-based loss functions for generative modeling. As shown in Lemma~\ref{lemma: kernel choice matters}, the choice of kernel is crucial: by emphasizing differences at critical frequencies in the spectral representation of MMD, an appropriate kernel can substantially influence training dynamics. This insight motivates adaptive kernel training. Nevertheless, kernel adaptivity is not a panacea and comes with its own limitations. In particular, we prove that there exist distributions for which MMD training fails for \emph{any} choice of kernel
\begin{lemma}[Worst-case distributions]
\label{lemma: mmd limitations}
There exist an absolute constant $b>0$ and an $n_0 \in \mathbb{N}$ such that, for every $n \geq n_0$, there are distributions $p$ and $q$ on \(\{0,1\}^n\) with
\begin{equation}
    \mathrm{TVD}(p,q)=1,
\end{equation}
and, for any bounded characteristic kernel with spectral measure $G$,
\begin{equation}
    \operatorname{MMD}_G(p,q) \le 2^{-b n}.
\end{equation}
\end{lemma}

\begin{proof}[Proof sketch]
The proof is based on the existence (we prove this rigorously in~\cref{proof: mmd limitations supplementary}) of distributions with exponentially decaying characteristic functions for all $\alpha \neq \mathbf{0}$: $\left|\phi_{p}(\alpha)\right| \leq 2^{-c_1 n}$ and $\left|\phi_{q}(\alpha)\right| \leq 2^{-c_2 n}$ for fixed $c_1, c_2 > 0$. Then, for any bounded characteristic kernel with spectral measure $G$, the corresponding MMD can be bounded above as
\begin{equation}
    \operatorname{MMD}_G^2\left(p, q\right) \leq \sup_{\alpha \neq \mathbf{0}} \left|\phi_{p}(\alpha) - \phi_{q}(\alpha)\right|^2 \leq 2^{-2n\min(c_1, c_2)}.
\end{equation}
\end{proof}
The above lemma also implies that for the worst-case distribution, the MMD's gradient w.r.t. the generator $\nabla_\theta \operatorname{MMD}_{G}^2$ is also exponentially vanishing ($\expval{Z_\alpha}_{q_\theta}$ is Lipschitz continuous w.r.t. to $\theta$), which again implies the kernel-based training becomes ineffective. This is a general limitation of the MMD metric, regardless of whether the generative model for training is classical or quantum. This issue is not unique to the IQP-QCBM model but rather a reflection of MMD's general insensitivity to certain distributions that can be artificially constructed. However, in practical scenarios, many real-world distributions can still be meaningfully distinguished using MMD with the adaptive training procedure.

\section{Numerical Experiments}\label{seq: numerics}
In this section, we present numerical experiments to evaluate the proposed kernel-adaptive training, employing synthetic datasets designed to illustrate scenarios where our method is beneficial.

\paragraph{Dataset} We construct each synthetic dataset $\mathcal{D}$ by drawing samples from the null space of a parity-check matrix $H \in \{0,1\}^{K \times n}$, i.e., for all $x \in \mathcal{D}$ and all $1 \leq k \leq K$, it holds that $\sum_{i=1}^n x_i h^{(k)}_i = 0 \pmod{2}$, where $h^{(k)}$ denotes the $k$-th row of $H$. The characteristic function of the underlying distribution $p$ of such a dataset differs from that of the uniform distribution $\mathcal{U}(\{0,1\}^n)$ only at certain frequencies, namely the vectors $h^{(k)}$ and their XOR combinations. We construct such datasets over $\{0,1\}^n$ with increasing $n = 12, 14, 16$, resulting in progressively harder tasks. Each dataset contains $2000$ training samples and $50000$ test samples.
To challenge Gaussian kernels, we construct datasets whose distributions have only three non-trivial frequencies at high Hamming weight, $\lvert \alpha \rvert > \frac{n}{2}$. It can be shown that, for any fixed bandwidth $\sigma \geq 0$, the squared Gaussian kernel MMD decays exponentially with the model size $n$. Assuming that at initialization the generated distribution $q_\theta$ is close to $\mathcal{U}(\{0,1\}^n)$, we expect both the magnitude of the Gaussian MMD loss and its gradient to vanish exponentially.

\begin{figure*}[tb]
\centering
\includegraphics[width=\textwidth]{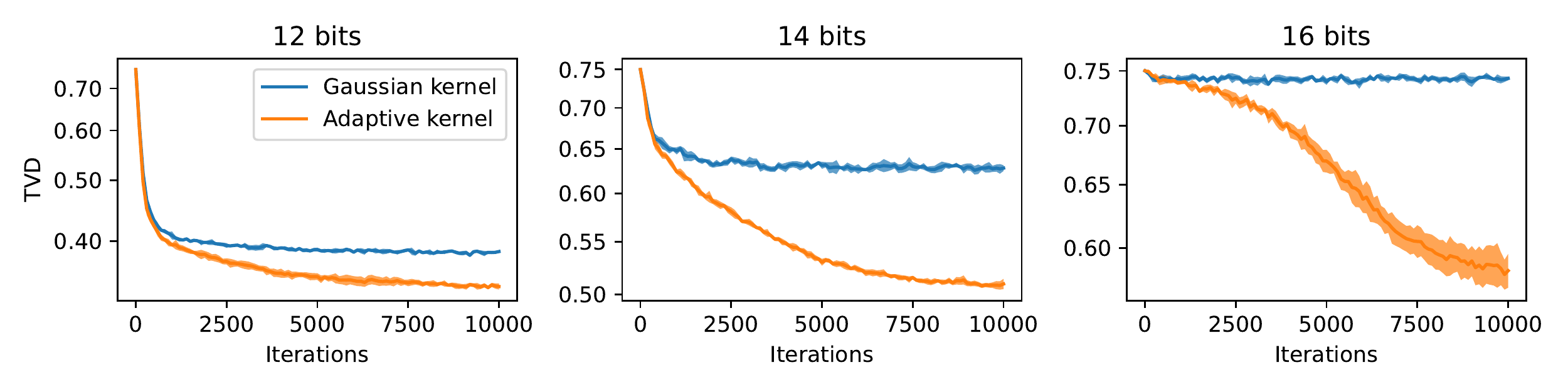}
\begin{tabular*}{0.9\linewidth}{@{\extracolsep{\fill}} lccc @{}}
    \toprule
    & \textit{12 bits} & \textit{14 bits} & \textit{16 bits} \\
    \midrule
    Gaussian kernel & 0.378$\pm$0.001 & 0.620$\pm$0.001 & 0.733$\pm$0.002 \\
    Adaptive kernel & \textbf{0.336$\pm$0.002} & \textbf{0.506$\pm$0.002} & \textbf{0.580$\pm$0.010} \\
    \bottomrule
  \end{tabular*}
\caption{Total variation distance between generator distribution and ground truth distribution, when trained with different kernels (mean $\pm$ standard deviation over $5$ runs) for 12-, 14-, and 16-bit synthetic parity-check datasets. Lowest achieved values are reported in the table. Bold indicates the best (minimal) TVD with statistical significance level of $0.01$.}
\label{fg: result}
\end{figure*}

\paragraph{Spectral Measure Parameterization} 
For our experiments, we employ a specific parametrization of the spectral measure $G_{\gamma}$ using a classic yet effective autoregressive model, the Fully Visible Sigmoid Belief Network (FVSBN) \cite{Neal1992, FreyDayanHinton1996}. More advanced architectures, such as NADE \cite{LarochelleMurray2011}, MADE \cite{Germain2015}, and Discrete Flows \cite{Tran2019DiscreteFlows}, could also be considered as alternatives. The FVSBN models a factorized joint distribution over binary variables as follows:
\begin{equation}
    G_{\gamma}({\alpha}) = \prod_{i=1}^n \bigl[p_{\gamma}(\alpha_i=1\mid{\alpha}_{<i})\bigr]^{\alpha_i} \bigl[1 - p_{\gamma}(\alpha_i=1\mid{\alpha}_{<i})\bigr]^{1 - \alpha_i}
\end{equation}
where the probability of each bit $\alpha_i$ is recursively computed from the previously sampled bits $\alpha_{<i}:=(\alpha_1,\dots,\alpha_{i-1})$. Specifically, the conditional probability is given by

\begin{equation}
    p_{\gamma}\!\left(\alpha_i = 1 \,\big|\, \alpha_{<i}\right) = \varepsilon + (1 - 2\varepsilon)\operatorname{sigmoid} \bigl(b_i + \sum_{r < i} W_{ir} (2\alpha_r - 1)\bigr)
\end{equation}
where the model is parameterized with parameters $\gamma = (W, b)$: $W\in\mathbb{R}^{n\times n}$ is a lower-triangular weight matrix, $b\in\mathbb{R}^n$ is a bias term. The small $\varepsilon=10^{-6}$ value is to ensure numerical stability. \\
FVSBN allows for both ancestral sampling of frequencies $\alpha$ and access to their log-probabilities $\log G_\gamma(\alpha) $, as required to compute the MMD loss gradients w.r.t. the parameters $\gamma$. Using the log-derivative trick, one can show that $\nabla_\gamma \widehat{\operatorname{MMD}}_{G_\gamma}^2(p, q_\theta)$ can be estimated as:
\begin{equation}
    \nabla_\gamma \widehat{\operatorname{MMD}}_{G_\gamma}^2(p, q_\theta) \approx \frac{1}{K}\sum_{k=1}^{K} \nabla_\gamma \log G_\gamma(\alpha_k) 
   \left(\langle Z_{\alpha_k} \rangle_p - \langle Z_{\alpha_k} \rangle_{q_\theta}\right)^2
\end{equation}
where $\alpha_k$ are frequencies sampled from $G_\gamma(\alpha)$. 

\paragraph{Results} For each synthetic parity-check dataset, we use the same IQP generator, which includes all gates acting on up to six qubits. Its parameters are initialized to match the data covariance, as described in \cite{recio2025train}. We compare our adaptive training to a fixed Gaussian kernel with bandwidth $\sigma=10^{-6}$ (approximately a uniform spectral measure). We tuned $\sigma$ as a hyperparameter and found that larger values degrade performance. For a fair comparison, $G_{\gamma}$ is initialized to the spectral measure of the Gaussian kernel with $\sigma=10^{-6}$. For the adaptive training and the fixed Gaussian kernel, we use the same IQP-QCBM generator. The model's performance is assessed by the total variation distance (TVD) observed on the test sample, measured in every 100 training iterations. 
The results are shown in \cref{fg: result}. We clearly see that for all three problem sizes, the adaptive training outperforms the Gaussian training. Furthermore, the performance gap widens at larger bit counts. For $16$ bits, while the Gaussian kernel training stagnates, the IQP generator improves and converges under adaptive training. 
As dimensionality increases, each minibatch covers an exponentially smaller fraction of the $2^{n}$ frequency space. Hence, a zero-initialised FVSBN draws the informative high-weight frequencies more rarely, which is visible in its slower convergence. 
To verify that our adaptive training generally outperforms Gaussian kernels regardless of the bandwidth, we repeated the training with different bandwidths $\sigma > 0$, or allowed the bandwidth to be dynamically updated during training. All these runs yield worse results than the $\sigma=0$ cases, which is expected from the special construction of the datasets. 

\section{Discussion and Conclusion}\label{seq: discussion}
In this work, we analyzed the prospects and limitations of the IQP-QCBMs as quantum generative models. In particular, we addressed two limitations. First, we established the universality of IQP-QCBM when hidden qubits are added, providing two proofs: a simple asymptotic version and an exact construction where adding $n + 1$ hidden qubits makes the $n$-qubit IQP generator universal. However, we emphasize that this construction serves primarily theoretical purposes and should not be interpreted as recommending large numbers of hidden units in practice. Our result instead illuminates hidden units and diagonal operations as architectural controls for balancing expressivity and trainability.

Second, we introduced an adversarial kernel-adaptive training method, which is classically efficient for IQP-QCBM. As we showed, this method addresses training issues in fixed-kernel MMD approaches. We provide theoretical guarantees that MMD convergence implies convergence in distribution under adaptive training. Experimental validation on the parity-check dataset demonstrates that this adaptive approach consistently reduces total variation distance compared to fixed Gaussian kernels, with performance gaps increasing as qubit count grows. While adaptive kernels can resolve training issues, we note that MMD-based training retains inherent limitations.

For future works, we plan to understand lower bounds on qubits required for hidden-model universality, implement the method for hidden-qubit architectures to quantify expressivity-trainability tradeoffs, and test kernel-adaptive training on more real-world datasets against fixed-kernel baselines.

\section{Acknowledgments}
AK thanks Yash Patel for insightful discussions. VD, HW, and AK acknowledge the support from the Dutch National Growth Fund (NGF), as part of the Quantum Delta NL programme. VD acknowledges support from the Dutch Research Council (NWO/OCW), as part of the Quantum Software Consortium programme (project number 024.003.03). This project was also co-funded by the European Union (ERC CoG, BeMAIQuantum, 101124342).

\bibliographystyle{alpha}  
\bibliography{refs}

\appendix
\section{Proof of \cref{lemma: asymptotic approximate universality}}\label{proof: asymptotic approximate universality}
\begin{proof}
After applying Hadamard gates to all $m+n$ qubits, followed by a diagonal phase operator $D(\theta)=\mathrm{diag}\!\big(e^{i\theta_{j,k}}\big)$ in the computational basis, the state we obtain is
\begin{equation*}
    \ket{\psi} = \frac{1}{\sqrt{2^m}} \sum_{j \in \{0,1\}^m} \ket{j}\,\ket{v(\theta_j)} .
\end{equation*}
where
\begin{equation*}
    \ket{v(\theta_j)} := \frac{1}{\sqrt{2^n}} \sum_{k \in \{0,1\}^n} e^{i\theta_{j,k}} \ket{k},
    \qquad
    \theta_j := \{\theta_{j,k}\}_{k \in \{0,1\}^n}.
\end{equation*}
By choosing $\theta_j \in \{0,\pi\}^n$, each $\ket{v(\theta_j)}$ can be made equal to any $X$-basis state $\ket{\tilde s}$. Consequently, applying another layer of Hadamards to the last $n$ qubits maps $\ket{\tilde s}$ to the corresponding computational-basis state $\ket{s}$; thus, by tuning the parameters we can obtain any $\ket{s}$ on those qubits. The marginal distribution on the last $n$ qubits is therefore the fraction of the $2^m$ hidden indices that output $s$:
\begin{equation*}
    q_\theta(s) = \frac{1}{2^m} \left| \left\{ j \in \{0,1\}^m \,\middle|\, H^{\otimes n}\ket{v(\theta_j)}=\ket{s} \right\} \right|.
\end{equation*}
Hence $q_\theta$ consists exactly of probabilities of the form $q_\theta(s)=c_s/2^m$ with nonnegative integers $c_s$ satisfying $\sum_s c_s=2^m$; conversely, any such choice is realizable by assigning exactly $c_s$ indices $j$ to output $s$.

It remains to bound the approximation error. Let $p=(p_1,\dots,p_{2^n})$, then we can always write each component as
\begin{equation*}
    p_j = \frac{i_j}{2^m} + \varepsilon_j,\qquad
i_j := \big\lfloor 2^m p_j \big\rfloor,\quad
\varepsilon_j \in \Big[0,\frac{1}{2^m}\Big).
\end{equation*}
Define the total rounding remainder $r := 2^m \sum_{j=1}^{2^n} \varepsilon_j \in \{0,1,\dots,2^n-1\}$. Form a valid grid distribution $q$ by adding $1/2^m$ to any $r$ components (e.g., the first $r$):
\begin{equation*}
    q_j =
    \begin{cases}
        \frac{i_j}{2^m} + \frac{1}{2^m}, & j\le r,\\
        \frac{i_j}{2^m}, & j>r.
    \end{cases}
\end{equation*}
Thus, we have $|p_j-q_j| \leq \frac{1}{2^m}$ for each $j$. So the total variation distance can be upper-bounded as follows:
\begin{equation*}
    \mathrm{TVD}(p,q) =\frac{1}{2} \sum_{j=1}^{2^n} |p_j-q_j| \leq \frac{1}{2} \cdot \frac{2^n}{2^m} \in\mathcal{O} \left(\frac{1}{2^{m-n}}\right).
\end{equation*}
Since every such $q$ is realizable as $q_\theta$ by the construction above, the claim follows.
\end{proof}

\section{Proof of \cref{thm: exact universality}}\label{proof: exact universality}
We prove several auxiliary lemmas first, before the proof of the main statement. 
\begin{lemma}\label{lemma: rho_b decomposition}
    Let the total quantum system consist of $m+n$ qubits, where the quantum state is prepared using an IQP circuit, specified by diagonal matrix $D(\theta) = \sum_{x \in \{0,1\}^m} \sum_{y \in \{0,1\}^n} e^{i\theta_{x, y}} |x\rangle \langle x| \otimes |y\rangle \langle y|$ with the final layer of Hadamard gates omitted. Then, the reduced density matrix $\rho_2$ of the second, $n$-qubit, subsystem is given by
    \begin{equation*}
        \rho_2 = \frac{1}{2^m} \sum_{k \in \{0,1\}^m} |\psi_k\rangle \langle \psi_k|,
    \end{equation*}
    where each state $|\psi_k\rangle$ takes the form
    \begin{equation*}
        |\psi_k\rangle = \frac{1}{\sqrt{2^n}} \sum_{y \in \{0,1\}^n} e^{i\theta_{k, y}} |y\rangle,\label{eq: psi}
    \end{equation*}
    and the phases $\theta_{k, y}$ are trainable parameters of the IQP circuit.
\end{lemma}

\begin{proof}
The density matrix of the whole system after the layer of Hadamards and the diagonal operator can be written as: 
    \begin{align*}
        \rho_{1+2} &= D(\theta)\, H^{\otimes(m+n)} \bigl(|0\rangle\!\langle 0|\bigr)^{\otimes(m+n)} H^{\otimes(m+n)} D^{\dagger}(\theta) \notag \\
        &= \frac{1}{2^{m+n}}
        \sum_{x,x'\in\{0,1\}^{n}} \sum_{y,y'\in\{0,1\}^{m}}
        e^{i\bigl(\theta_{x,y}-\theta_{x',y'}\bigr)}
        \bigl(|x\rangle\!\langle x'| \otimes |y\rangle\!\langle y'|\bigr).
    \end{align*}
Reduced density matrix for the second subsystem:
\begin{align*}
    \rho_2 &= \sum_{k \in \{0,1\}^n} \frac{1}{2^{m+n}} \sum_{x,x'\in \{0,1\}^m} \sum_{y,y' \in \{0,1\}^n} e^{i(\theta_{x,y} - \theta_{x',y'})} \langle k|x\rangle \langle x'|k\rangle \otimes |y\rangle \langle y'| = \notag \\
    &= \sum_{y,y'\in \{0,1\}^n} \left(\frac{1}{2^{m+n}} \sum_{k \in \{0,1\}^m} e^{i(\theta_{k,y} - \theta_{k,y'})} \right) |y\rangle \langle y'| = \notag \\&= \frac{1}{2^m}\sum_{k \in \{0,1\}^m} \underbrace{\left(\frac{1}{\sqrt{2^n}}\sum_{y\in \{0,1\}^n} e^{i\theta_{k,y}}|y\rangle \right)}_{|\psi_k\rangle} \left(\frac{1}{\sqrt{2^n}}\sum_{y'\in \{0,1\}^n} e^{-i\theta_{k,y'}}\langle y'| \right) \notag \\&=  \frac{1}{2^m} \sum_{k \in \{0,1\}^m} |\psi_k\rangle \langle \psi_k|
\end{align*}
\end{proof}

\begin{lemma}[2-sparse distributions via IQP state]
\label{lemma:2-sparse-distribution-encoding}
Let $\tilde s_1,\tilde s_2$ be distinct $n$-qubit $X$-basis product states and let $p\in[0,1]$.
Then there exist phases $\{\theta_y\}_{y\in\{0,1\}^n}$ such that the IQP state $\ket{\psi} =\frac{1}{\sqrt{2^n}}\sum_{y\in\{0,1\}^n} e^{i\theta_y} \ket{y}$ satisfies
\begin{equation*}
    \left|\bra{\tilde s_1}\ket{\psi}\right|^2 = p, \qquad \left|\bra{\tilde s_2}\ket{\psi}\right|^2 = 1 - p.
\end{equation*}
\end{lemma}

\begin{proof}
A state is \emph{uniform-magnitude (UMA)} in the computational basis iff it can be written as
$\frac{1}{\sqrt{2^n}}\sum_{y} e^{i\theta_y}\ket{y}$ for some phases $\{\theta_y\}$, so it suffices to
show that for some $\gamma_1,\gamma_2$ the state
\begin{equation}\label{eq: 2 sparse state}
    \sqrt{p}e^{i\gamma_1}\ket{\tilde s_1}+\sqrt{1-p}e^{i\gamma_2}\ket{\tilde s_2}
\end{equation}
is UMA.

Start with $n=1$. The UMA state $\tfrac{1}{\sqrt2}(\ket{0}+e^{i\theta}\ket{1})$
equals $e^{i\theta/2}\big(\cos(\tfrac\theta2)\ket{+}-i\sin(\tfrac\theta2)\ket{-}\big)$, so choosing
$\theta$ with $\cos^2(\tfrac\theta2)=p$ yields amplitudes $\sqrt{p},\sqrt{1-p}$ (the phases are absorbed
into $\gamma_1,\gamma_2$). This covers $p\in\{0,1\}$ via $\theta\in\{0,\pi\}$.

Consider now the arbitrary $n$. Pick a position where $\tilde s_1$ and $\tilde s_2$ differ and relabel it as qubit~1
(permutations preserve UMA). Then
$\ket{\tilde s_1}=\ket{+}\ket{\tilde s'_1}$ and $\ket{\tilde s_2}=\ket{-}\ket{\tilde s'_2}$ for some
$(n-1)$-qubit $X$-basis product states $\ket{\tilde s'_1},\ket{\tilde s'_2}$. If
$\ket{\tilde s'_1}=\ket{\tilde s'_2}$, we are done since the tensor product of UMA states is UMA.

Otherwise, choose the diagonal $(n\!-\!1)$-qubit unitary $U=\bigotimes_{j=2}^n Z^{b_j}$ (with
$b_j\in\{0,1\}$ chosen so that $U\ket{\tilde s'_1}=\ket{\tilde s'_2}$; note $Z$ toggles
$\ket{+}\!\leftrightarrow\!\ket{-}$). Define the $X$-basis–controlled unitary
$\mathrm{hc}U := \ket{+}\!\bra{+}\otimes I + \ket{-}\!\bra{-}\otimes U$. Then
\begin{equation*}
    \mathrm{hc}U\!\left(\big(\sqrt{p}e^{i\gamma_1}\ket{+}+\sqrt{1-p}e^{i\gamma_2}\ket{-}\big)\otimes\ket{\tilde s'_1}\right)
= \sqrt{p}e^{i\gamma_1}\ket{\tilde s_1}+\sqrt{1-p}e^{i\gamma_2}\ket{\tilde s_2}.
\end{equation*}
In the computational basis $\mathrm{hc}U$ is a block matrix
\begin{equation*}
    \mathrm{hc}U=\begin{pmatrix}(I+U)/2 & (I-U)/2\\[2pt] (I-U)/2 & (I+U)/2\end{pmatrix},
\end{equation*}
Note that in each row, we have only one non-zero element. So this is a permutation matrix. For a permutation matrix $U$, it holds that $\ket{\phi}$ is UMA if and only if $U\ket{\phi}$ is UMA. Thus, \cref{eq: 2 sparse state} represents the UMA state, which concludes the proof.
\end{proof}

Next, we present the two key probability distribution decomposition lemmas in the language of probability vectors, i.e., vectors $p=(p_j)_j$, $\sum_j p_j = 1,\ p_j\geq 0$, encoding the probability of the measurement of each of the bitstrings (indexed by $j$).  
We will say a (probability) vector is $k$-sparse if it has at most $k$ non-zero entries.

\begin{lemma}
Every $N$-dimensional probability vector $p$ can be expressed as a uniform mixture of $N$ 3-sparse probability vectors. That is, there exists a set $\{ q(i)\}_i$ of $N$ 3-sparse probability vectors such that
$$
p = \sum_i \frac{1}{N} q(i).
$$
\end{lemma}

\textit{Proof.} We give a constructive proof by showing there exists an allocation matrix $Q\in\mathbb{R}^{N\times N}$ of $p$ such that the column sum $\sum_i Q_{ij} = p_j$ and the row sum $\sum_j Q_{ij} = 1/N$, which is also 3-row-sparse. The 3-sparse probability vectors are the rows of the $Q$, i.e., $q(i)/N$ is the $i$th row of $Q$.
We assume $p_1\leq p_2 \leq \cdots \leq p_N$ w.l.o.g. Let $k$ be the smallest index such that $p_{k+1} \geq 1/N$. Note $k=1$ if the distribution is uniform and can be as large as $N-1$, but not $N$. \\
The general idea is that we first allocate $p_1,\ldots, p_k$ to the first $k$ diagonal entries of $Q$. Since $p_1\leq p_2\leq\cdots \leq p_k \leq 1/N$, we have some capacity left on each row, and hence we can further try to spread out $p_{k+1}$ to the first $k$ rows. After $p_{k+1}$ is fully spread out, we proceed with $p_{k+2}$ and so on. Note that whenever the first rows are out of capacity, we will proceed using the remaining $N-k$ rows. The matrix $Q$ can be constructed and verified with the following three steps. \\
Step 1: Initialization.
$$Q = 
\begin{bmatrix}
Q_0 &0 \\
0 & 0
\end{bmatrix}
,\;
Q_0 = 
\begin{bmatrix}
p_1 & 0   & \cdots & 0 \\
0   & p_2 & \cdots & 0 \\
\vdots & \vdots & \ddots & \vdots \\
0   & 0   & \cdots & p_k
\end{bmatrix}
$$
\\
Step 2: Iterate over $\ell=k+1, \ldots, N$. For each iteration, we spread out $p_\ell$ to the $\ell$ column of $Q$ as follows. For $i$ running from 1 to $N$, we set
\begin{equation}\label{eq:Q_update}
    Q_{i\ell} = \min\bigg\{\frac{1}{N} - \sum_{j}Q_{ij}, p_\ell - \sum_{\alpha} Q_{\alpha \ell}\bigg\}.
\end{equation}
Note that, $N^{-1} - \sum_{j}Q_{ij}$ is the remaining capacity of row $i$ (recall each row sums to $1/N$) and $p_\ell - \sum_{\alpha < i} Q_{\alpha \ell}$ is the residual of $p_\ell$ after spreading it over the first $i-1$ row. Two indicator functions check if a row is out of capacity or $p_\ell$ is completely spread out. \\
Step 3: Correctness. After iteration $\ell = k+1$, we must have at most two nonzero entries in each row since initially, $Q$ has at most one nonzero value per each row, and by Eq.~\eqref{eq:Q_update}, we only modify one entry of each row. Let $r$ be the largest row index such that $Q_{r\ell}\neq 0, \ell = k+1$. After the next iteration $\ell = k+2$, we have two cases: (1) if the capacity of row $r$ is used up by $p_{k+1}$, then the first nonzero entry of column $k+2$ starts at row $r$. In this case, $Q$ has at most two nonzero entries per row; (2) if some capacity of row $r$ is left, i.e., $\sum_{j} Q_{rj} \leq 1/N$, then first nonzero entry of column $k+2$ starts at row $r+1$, giving rise to three nonzero entries for this row. Note that the above argument between iterations $k+1$ and $k+2$ holds for any two iterations $\ell$ and $\ell+1$. Hence, we conclude that after iteration $\ell \geq k+1$, the following conditions are true:
\begin{enumerate}
    \item $\sum_{i} Q_{ij} = p_j$ for $1 \leq j \leq \ell$;
    \item There are at most three nonzero entries in each row of $Q$.
\end{enumerate} 
Now, proceed with the above argument until $\ell = N$, we must have: (1) $\sum_{j} Q_{ij} = 1/N$ for $1\leq i \leq N$; (2) $\sum_{i} Q_{ij} = p_j$ for $1 \leq j \leq N$; (3) there are at most three nonzero entries in each row of $Q$.
\qed

Next, we show that using twice as many $q$-vectors, we can represent the probability vector $p$ with 2-sparse probabilities. The basic idea is that a 3-sparse vector can always be written as a uniform mixture of two 2-sparse vectors: Let $p$ be a 3-sparse vector with $0 < p_a \leq p_b \leq p_c$, for indices $a, b, c$. Then we define the probability vectors $q_1$ and $q_2$ as follows: $(q_1)_a = 2p_a$, $(q_1)_c = 1 - 2 p_a$, $(q_2)_b = 2 p_b$, and $(q_2)_c = 1 - 2p_b$. The correctness can be verified by $q_1,q_2$ are 2-sparse, $q_1/2 + q_2/2 = p$, and $p_a\leq 1/2$ and $p_b \leq 1/2$.

\begin{lemma}\label{lemma: any distributin with 2-sparse}
Every $N$-dimensional probability vector $p$ can be expressed as a uniform mixture of~$2N$ 2-sparse probability vectors. That is, there exists a set $\{ q(i) \}_i$ of $2N$ 2-sparse probability vectors such that
$$
p = \sum_i \frac{1}{2N} q(i).
$$
\end{lemma}
\textit{Proof.} Taking Lemma 5, we know there exists a set $\{ q(i)\}_i$ of $N$ 3-sparse probability vectors such that $p = \sum_i \frac{1}{N} q(i)$. Now, for each $q(i)$, if it is 3-sparse but not 2-sparse, we can always split it into two 2-sparse vectors using the observation before the lemma, i.e., $q(i) = q_1(i)/2 + q_2(i)/2$; If $q(i)$ is  2-sparse, we simply produce two copies of it.
This construction satisfies the conditions of the lemma.  
\qed

With these lemmas in place, we can now prove the main result summarised in \cref{thm: exact universality}:
\begin{proof}
By Lemma~\ref{lemma: rho_b decomposition} the reduced density matrix of an IQP circuit with $m$ hidden qubits can be exactly written as sum of $2^m$ pure states density matrices $|\psi_k\rangle\langle\psi_k|$.
The output distribution will be a uniform mixture of the distributions obtained from measuring one of the $|\psi_k\rangle$ states in the basis specified by Pauli-X eigenstates. 
Each $|\psi_k\rangle$ can be expressed as a superposition of strings of Pauli-X eigenstates, and the probability of observing the string $b \in \{0,1\}^n$ is given by $|\langle \tilde{b}| \psi_k\rangle|^2$
with $|\tilde{b}\rangle = H^{\otimes n} |b\rangle  $.
In Lemma~\ref{lemma: any distributin with 2-sparse} we have shown how we can encode any 2-sparse distribution in  $|\psi_k\rangle$.
 On the other hand, Lemma~\ref{lemma: any distributin with 2-sparse} states that any probability distribution over $\{0,1\}^n$ can be decomposed as a uniform sum of $2^{n+1}$ such 2-sparse distributions. Therefore, by choosing $m = n+1$, we have sufficient expressive power to represent any target distribution $p$ exactly using an IQP circuit with hidden qubits.
\end{proof}

\section{(Non-)Universality: two qubits model}\label{proofs: Hidden v.s. visible only: toy example}
\begin{lemma}[distributions beyond $2$-qubit model expressive power]\label{lemma: worst case distributions for 2-qubit model}
    Any $3$-sparse distribution over $\{0, 1\}^2$ with exactly $3$ non-zero components cannot be represented by a $2$-qubit IQP-QCBM.
\end{lemma}

\begin{proof}
    It is equivalent to show that any $3$-sparse distribution over $\{+, -\}^2$ with exactly $3$ non-zero components cannot be represented by a $2$-qubit IQP-QCBM with the final Hadamard layer omitted.
    W.l.o.g., consider a target distribution with $p_{++}, p_{+-}, p_{-+} > 0$ and $p_{--} = 0$. The model produces distributions of the form
    \begin{equation}\label{eq: probs produced by 2-qubit model}
        p_{\sigma_1, \sigma_2} = \frac{1}{16}\bigl|1 + \sigma_1 e^{i\theta_1} + \sigma_2 e^{i\theta_2} + \sigma_1 \sigma_2 e^{i\theta_3}\bigr|^2,
    \end{equation}
    with $\sigma_1, \sigma_2 \in \{+,-\}$.
    From $p_{--} = 0$ we get
    \begin{equation}
        0=\frac{1}{16}\bigl|1 - e^{i\theta_1} - e^{i\theta_2} + e^{i\theta_3}\bigr|^2 \Rightarrow e^{i\theta_3} = e^{i\theta_1} + e^{i\theta_2} - 1.
    \end{equation}
    Substituting this into \cref{eq: probs produced by 2-qubit model} gives, and writing the normalization condition gives $\sum_{\sigma_1,\sigma_2} p_{\sigma_1,\sigma_2}=1$ is equivalent here to
    \begin{equation}
      \cos^2\!\left(\frac{\theta_1-\theta_2}{2}\right)
      + \sin^2\!\left(\frac{\theta_1}{2}\right)
      + \sin^2\!\left(\frac{\theta_2}{2}\right) = 1,
    \end{equation}
    This yields the conditions
    \begin{equation}
        \theta_1-\theta_2=\pi \ (\mathrm{mod}\ 2\pi) \quad\text{or}\quad \theta_1=0 \ (\mathrm{mod}\ 2\pi) \quad\text{or}\quad \theta_2=0 \ (\mathrm{mod}\ 2\pi).
    \end{equation}
    In each case at least one of $p_{++},p_{+-},p_{-+}$ vanishes:
    if $\theta_1-\theta_2=\pi$ then $p_{++}=0$; if $\theta_1=0$ then $p_{-+}=0$; if $\theta_2=0$ then $p_{+-}=0$.
    This contradicts the assumption that all three are strictly positive, completing the proof.
\end{proof}

\begin{lemma}\label{lem:4sparse_hidden_simple}
Every probability distribution on $\{0,1\}^2$ can be represented by a \(2\)-qubit IQP--QCBM with \emph{one} hidden qubit.
\end{lemma}

\begin{proof}
Similar to proof of \cref{lemma: worst case distributions for 2-qubit model} we ignore last layer of Hadamards by switching to distributions over $\{+, -\}^2$ representation. Recall that distribution over $\{+, -\}^2$ produced by $2$-qubit IQP-QCBM (omitting last Hadamard layer) can be represented by choosing the parameters in:
    \begin{equation}
        p_{\sigma_1, \sigma_2} = \frac{1}{16}|1 + \sigma_1 e^{i\theta_1} + \sigma_2 e^{i\theta_2} + \sigma_1 \sigma_2 e^{i\theta_3}|^2
    \end{equation}
    with $\sigma_1, \sigma_2 \in \{+, -\}$. We split the proof into steps for easier understanding.
    
\emph{Step 1: Geometric interpretation of model and target distribution.} First, we can simplify the distribution above produced by model without hidden qubits:
\begin{align*}
    p_{\sigma_1, \sigma_2} &= \frac{(1 + \sigma_1 e^{i\theta_1} + \sigma_2 e^{i\theta_2} + \sigma_1 \sigma_2 e^{i\theta_3})(1 + \sigma_1 e^{-i\theta_1} + \sigma_2 e^{-i\theta_2} + \sigma_1 \sigma_2 e^{-i\theta_3})}{16} \\ &= \frac{1}{4}(1 + \sigma_1 x + \sigma_2 y + \sigma_1 \sigma_2 z)
\end{align*}
where we denote
\begin{equation}\label{eq: iqp produced inside thetrahedron}
    \begin{cases}
        x &=\tfrac12\big(\cos\theta_1+\cos(\theta_2-\theta_3)\big),\\
        y &=\tfrac12\big(\cos\theta_2+\cos(\theta_1-\theta_3)\big),\\
        z &=\tfrac12\big(\cos\theta_3+\cos(\theta_1-\theta_2)\big).
    \end{cases}
\end{equation}
Normalization is automatic. Non-negativity is equivalent to that $(x,y,z)$ lies in the tetrahedron:
\begin{equation}
   T = \{(x, y, z)| - 1 + |x + y| \leq z \leq 1 - |x - y|\}
\end{equation}
which is specified by the vertices:
\begin{equation}
    \{(1, 1, 1), (1, -1, -1), (-1, 1, -1), (-1, -1, 1)\}
\end{equation}

Thus, the interpretation is the following. An arbitrary target probability distribution is a point that belongs to the tetrahedron $T$, but the parametrized model without hidden qubits only reaches a subset of the tetrahedron that satisfies \cref{eq: iqp produced inside thetrahedron}.

\emph{Step 2: Adding one qubit.} To simplify our calculations further, we purposefully restrict even more \cref{eq: iqp produced inside thetrahedron}. Consider $\theta_3=\theta_1+\theta_2+\delta$ with $\delta\in[0,\pi]$, this leads to simple parametrization of visible only qubit 

\begin{equation}
    \begin{cases}
        x &= au, \\
        y &= av, \\
        z &= uv,
    \end{cases}
    \qquad a \in [0,1], \\
    \qquad u,v \in [-1,1].
\end{equation}
where the following notations used:
\begin{equation}
a:=\cos(\delta/2)\in[0,1],\quad
u:=\cos(\theta_1+\tfrac{\delta}{2}),\ v:=\cos(\theta_2+\tfrac{\delta}{2})\in[-1,1].
\end{equation}

With one hidden qubit, we roughly speaking average two circuits with visible qubits only with independent parameters, which follows directly from \cref{lemma: rho_b decomposition}, under simplification above it can be written as: 
\begin{equation}\label{eq:avg_general}
    \begin{cases}
        x &= \tfrac12(au+\tilde a\tilde u), \\
        y &= \tfrac12(av+\tilde a \tilde v), \\
        z &= \tfrac12(uv+\tilde u\,\tilde v),
    \end{cases}
    \qquad a, \tilde a \in [0,1], \\
    \qquad u,v,\tilde u,\tilde v\in[-1,1].
\end{equation}

\emph{Step 3: Hitting an arbitrary target $(x,y,z)\in T$.}
By symmetry, we may assume $x\ge y\ge0$.

\emph{Case A: $x>\tfrac12$.}
Set $a=1$ and $u=\tilde v=1$. Then
\begin{equation}
    \tilde u=\frac{2x-1}{\tilde a},\qquad v=2y-\tilde a,\qquad \tilde a\in[2x-1, 1],
\end{equation}
and
\begin{equation}
    z(\tilde a)=\frac{2y-\tilde a}{2}+\frac{2x-1}{2\tilde a},\quad z(1)=x+y-1,\quad z(2x-1)=1-(x-y).
\end{equation}
By continuity, $z$ sweeps the full band $[x+y-1,\ 1-(x-y)]$.

\emph{Case B: $x\le\tfrac12$.}
Fix
\begin{equation}
    a:=1-2x,\quad \tilde a:=1,\quad u:=-1,\quad \tilde u:=1,
\end{equation}
let $v\in[-1,1]$ be free, and set $\tilde v:=2y-(1-2x)v$. Then
\begin{equation}
    x=\tfrac12(au+\tilde a \tilde u)=\tfrac12\big(-(1-2x)+1\big)=x,\quad y=\tfrac12(av+\tilde a \tilde v)=y,
\end{equation}
and
\begin{equation}
    z=\tfrac12(uv+\tilde u\tilde v)=y-(1-x)v.
\end{equation}
As $v$ ranges over $[-1,1]$, we obtain
\begin{equation}
    z\in[y-(1-x),\ y+(1-x)]=[x+y-1,\ 1-(x-y)].
\end{equation}

Feasibility: $u,\tilde u\in[-1,1]$, and since $0\le y\le x\le\tfrac12$ we have
$2y\pm(1-2x)\in[-1,1]$, so $\tilde v\in[-1,1]$ for all $v\in[-1,1]$.

Combining Cases A and B (and using sign flips/swaps for other quadrants) shows that, for any fixed $(x,y)$, the reachable $z$ coincides with the entire tetrahedron interval $[-1+|x+y|,\ 1-|x-y|]$.

\end{proof}
\section{Proof of \cref{remark: efficient Pauli-Z estimation IQP with hidden qubits}}\label{proof: efficient Pauli-Z estimation IQP with hidden qubits}

\begin{proof}
    The key idea is to express the quantum expectation $\langle Z_\alpha \rangle_{q_\theta}$ as a classical expectation over uniformly random bitstrings, which can be estimated efficiently using Monte Carlo sampling similar to IQP-QCBMs without hidden qubits proven in \cite{recio2025train}.

    Due to the structure of IQP-QCBM, the expectation $\langle Z_\alpha \rangle_{q_\theta}$ can be re-expressed as an expectation over uniformly random bitstrings:
    \begin{equation}
        \langle Z_\alpha \rangle_{q_\theta} = \mathbb{E}_{y \sim U_m} \mathbb{E}_{z \sim U_n} \Bigg[\cos \Bigg( \sum_j \theta_j (-1)^{g_j^{(y)} \cdot y \oplus g_j^{(z)} \cdot z}  \left( 1 - (-1)^{z \cdot \alpha} \right) \Bigg) \Bigg]
    \end{equation}
    where $y, z$ is sampled uniformly at random. The vectors $g_j^{(y)}$ and $g_j^{(z)}$ denote the $g_j$ acting on hidden and visible qubits, respectively.
    
    This expectation can be approximated by randomly sampling pairs $(y_k, z_i)$ and computing:
    \begin{equation}
        \widehat{\langle Z_\alpha \rangle}_{q_\theta} = \frac{1}{|Y|\, |Z|} \sum_{k=1}^{|Y|} \sum_{i=1}^{|Z|} \cos \Bigg( \sum_j \theta_j (-1)^{g_j^{(y)} \cdot y_k \oplus g_j^{(z)} \cdot z_i} \left( 1 - (-1)^{z_i \cdot \alpha} \right) \Bigg)
    \end{equation}
    By construction, the estimator is unbiased:
    \begin{equation}
        \mathbb{E}_{Y,Z} \left[ \widehat{\langle Z_\alpha \rangle}_{q_\theta} \right] = \langle Z_\alpha \rangle_{q_\theta}
    \end{equation}
    To achieve an additive error $\varepsilon$, it suffices to take
    \begin{equation}
        |Y||Z| = \mathcal{O}\left(\frac{1}{\varepsilon^2}\right)
    \end{equation}
    samples. Setting $\varepsilon = \mathcal{O}(\mathrm{poly}(n^{-1}, m^{-1}))$ requires only a polynomial number of samples in $n$ and $m$. Thus, even in the presence of hidden qubits, the expectation value $\langle Z_\alpha \rangle_{q_\theta}$ can be estimated classically to polynomial accuracy using a Monte Carlo method with polynomial runtime.
    \end{proof}

\section{Proof of \cref{lemma: kernel choice matters}}\label{proof: kernel choice matters}
\begin{proof}
Let the target distribution be the uniform distribution:
\begin{equation}
p(x) = \frac{1}{2^n}, \quad \forall x \in \{0,1\}^n.
\end{equation}

Consider an IQP-QCBM circuit with generators $g_j = (0,\ldots,0,1,0,\ldots,0)$ (a single-qubit $Z$ term).  
For parameters $\theta' = (\pi/8, \pi/4, \ldots, \pi/4)$, the circuit outputs:
\begin{equation}
q_{\theta'}(x) =
\begin{cases}
\dfrac{\cos^2(\pi/8)}{2^{n-1}}, & x_1 = 0, \\[6pt]
\dfrac{\sin^2(\pi/8)}{2^{n-1}}, & x_1 = 1.
\end{cases}
\end{equation}

The TV distance simplifies to:
\begin{align}
    \operatorname{TV}(p, q_{\theta'}) 
    &= \frac{1}{2} \sum_{x \in \{0,1\}^n} |p(x) - q_{\theta'}(x)| = \frac{1}{2} \Bigg(
        2^{n-1} \left| \frac{\cos^2\left(\frac{\pi}{8}\right)}{2^{n-1}} - \frac{1}{2^n} \right| 
      + 2^{n-1} \left| \frac{\sin^2\left(\frac{\pi}{8}\right)}{2^{n-1}} - \frac{1}{2^n} \right|
    \Bigg) \nonumber \\
    &= \frac{1}{4} \Big( |2\cos^2(\frac{\pi}{8}) - 1| + |2\sin^2(\frac{\pi}{8}) - 1| \Big)
      = \frac{2}{4} \cos\left(\frac{\pi}{4}\right)
      = \frac{\sqrt{2}}{4}
\end{align}
This is a constant independent of $n$.

The Fourier characteristic for target and model distribution can be written as follows:
\begin{equation}
\phi_p(\alpha) = \delta_{\alpha, 0}, 
\qquad 
\phi_{q_{\theta}}(\alpha) = \prod_{j: \alpha_j=1} \cos(2\theta_j).
\end{equation}
At $\theta'$:
\begin{equation}
\phi_{q_{\theta'}}(\alpha) =
\begin{cases}
1, & \alpha = (0,\dots,0), \\[4pt]
\frac{\sqrt{2}}{2}, & \alpha = (1,0,\dots,0), \\[4pt]
0, & \text{otherwise}.
\end{cases}
\end{equation}
Thus, the only nonzero discrepancy occurs at $\alpha^\star = (1,0,\dots,0)$:
\begin{equation}
|\phi_p(\alpha) - \phi_{q_{\theta'}}(\alpha)| =
\begin{cases}
\dfrac{\sqrt{2}}{2}, & \alpha = \alpha^\star, \\
0, & \text{otherwise}.
\end{cases}
\end{equation}
The partial derivatives of model distribution w.r.t. parameter $\theta_k$:
\begin{equation}
    \frac{\partial}{\partial \theta_k}\phi_{q_{\theta'}}(\alpha) = 
    \begin{cases}
        -2 \sin(2\theta_k) \prod_{j\neq k, j: \alpha_j = 1} \cos(2\theta_j)& \text{if } \alpha_k =1 \\
        0 & \text{otherwise}
    \end{cases}
\end{equation}

For any bounded characteristic kernel $k$ with spectral measure $G(\alpha)$:
\begin{equation}
\operatorname{MMD}^2_{k}(p, q_{\theta'}) 
= \sum_{\alpha \in \{0,1\}^n} 
|\phi_p(\alpha) - \phi_{q_{\theta'}}(\alpha)|^2 \, G(\alpha) 
= \tfrac{1}{2} G(\alpha^\star).
\end{equation}
Gradient structure:
\begin{align}
\left\| \nabla_\theta \operatorname{MMD}^2_{k_\sigma}(p, q_{\theta'}) \right\| 
&= \sqrt{\sum_{j=1}^n \left( \frac{\partial}{\partial \theta_j} 
\operatorname{MMD}^2_{k_\sigma}(p, q_{\theta'}) \right)^2 } \notag \\
&= \sqrt{\sum_{j=1}^n \left( 
\sum_{\alpha \in \{0,1\}^n} 
2 \, \big| \phi_p(\alpha) - \phi_{q_{\theta'}}(\alpha) \big| 
\, G(\alpha)  \frac{\partial}{\partial \theta_j} \phi_{q_{\theta}}(\alpha) 
\Bigg|_{\theta = \theta'} 
\right)^2 } \notag \\
&= 2 \, \big| \phi_p(\alpha^\star) - \phi_{q_{\theta'}}(\alpha^\star) \big| 
\, G(\alpha^\star) \left|\frac{\partial}{\partial \theta_1} \phi_{q_{\theta}}(\alpha^\star) 
\Bigg|_{\theta = \theta'}\right|= 2 G(\alpha^\star) 
\end{align}

\begin{itemize}
        \item \textit{Gaussian kernel.} The spectral the spectral measure admits:
        \begin{equation}
            G_\sigma(\alpha) = p_\sigma^{|\alpha|}(1-p_\sigma)^{n-|\alpha|}, 
            \quad p_\sigma = \tfrac{1 - e^{-1/(2\sigma)}}{2} < \tfrac{1}{2}.
        \end{equation}
        This gives us an expression for MMD and its gradient norm:
        \begin{align}
            \operatorname{MMD}^2_{k_\sigma}(p, q_{\theta'}) 
            &= \tfrac{1}{2} p_\sigma (1-p_\sigma)^{n-1} \in \mathcal{O}(2^{-n}), \\[6pt]
            \big\| \nabla_\theta \operatorname{MMD}^2_{k_\sigma}(p, q_{\theta'}) \big\|
            &= 2 p_\sigma (1-p_\sigma)^{n-1} \in \mathcal{O}(2^{-n}).
        \end{align}
    \item \textit{Smart kernel choice.} Define a kernel with spectral measure concentrated on $\alpha^\star$:
    \begin{equation}
        G_\kappa(\alpha) =
        \begin{cases}
        1 - \varepsilon, & \alpha = \alpha^\star, \\[4pt]
        \dfrac{\varepsilon}{2^n - 1}, & \text{otherwise},
        \end{cases}
        \quad 0 < \varepsilon \ll 1.
    \end{equation}
    This kernel is characteristic since $G_\kappa(\alpha) > 0$ for all $\alpha$. For MMD and its gradient norm, we have:
    \begin{align}
        \operatorname{MMD}^2_\kappa(p, q_{\theta'}) 
        &= \tfrac{1}{2}(1-\varepsilon), \\[6pt]
        \big\| \nabla_\theta \operatorname{MMD}^2_\kappa(p, q_{\theta'}) \big\|
        &= 2(1-\varepsilon).
    \end{align}
    Thus, setting $C = \tfrac{1}{2}(1-\varepsilon)$ gives the desired constant lower bounds, which finishes the proof.
\end{itemize}
\end{proof}

\section{Proof of \cref{lemma: consistency with weak convergence}}
\begin{proof}\label{proof: consistency with weak convergence}
We prove both directions separately.

\noindent \textbf{($\Leftarrow$)} Suppose $q_t \xrightarrow{\mathrm{d}} p$. Since $\{0,1\}^n$ is finite, weak convergence is equivalent to convergence in total variation:
\begin{equation}
\lim_{t \to \infty} \mathrm{TVD}(q_t, p) = \frac{1}{2} \sum_{x \in \{0,1\}^n} |q_t(x) - p(x)| = 0.
\end{equation}

By Parseval's identity for the Fourier basis on $\{0,1\}^n$ and Cauchy-Schwarz inequality, we have:
\begin{align}
    \sum_{\alpha \in \{0,1\}^n} |\phi_p(\alpha) - \phi_{q_t}(\alpha)|^2 &= 2^n\sum_{x \in \{0,1\}^n} |p(x) - q_t(x)|^2 \leq 2^n \left( \sum_{x} |p(x) - q_t(x)| \right)^2 \notag\\
    &= 2^n \cdot 4\mathrm{TVD}(p, q_t)^2   
\end{align}

Therefore,
\begin{equation}
    0 \leq \lim_{t \to \infty} \mathcal{L}(q_t, p) \leq \lim_{t \to \infty} \left(4\cdot 2^n\mathrm{TVD}(q_t, p)^2 \max_{\gamma, \alpha} G_\gamma(\alpha)\right) = 0
\end{equation}

\noindent \textbf{($\Rightarrow$)} Conversely, suppose $\lim_{t \to \infty} \mathcal{L}(q_t, p) = 0$. 

By definition,
\begin{equation}
\mathcal{L}(q_t, p) = \max_\gamma \sum_{\alpha \in \{0,1\}^n} G_\gamma(\alpha) |\phi_{q_t}(\alpha) - \phi_p(\alpha)|^2.
\end{equation}

Since each $G_\gamma$ has full support, there exists a constant $c > 0$ such that $G_\gamma(\alpha) \geq c$ for all $\gamma$ and all $\alpha$. Thus,
\begin{equation}
    \mathcal{L}(q_t, p) \geq c \sum_{\alpha \in \{0,1\}^n} |\phi_{q_t}(\alpha) - \phi_p(\alpha)|^2 = 2^n c\sum_{x \in \{0,1\}^n} |q_t(x) - p(x)|^2 \to 0
\end{equation}

where the last equality follows from Parseval's identity.

Applying Cauchy-Schwarz inequality, we obtain
\begin{equation}
    \mathrm{TVD}(q_t, p) = \frac{1}{2} \sum_{x \in \{0,1\}^n} |q_t(x) - p(x)| \leq \frac{1}{2} \sqrt{2^n \sum_{x} |q_t(x) - p(x)|^2} \to 0
\end{equation}

Hence, $\mathrm{TVD}(q_t, p) \to 0$, which is equivalent to $q_t \xrightarrow{\mathrm{d}} p$ since $\{0,1\}^n$ is finite.
\end{proof}

\section{Supplementary proof for \cref{lemma: mmd limitations}}\label{proof: mmd limitations supplementary}
\begin{lemma}\label{lemma: exponentially decaying characteristic functions}
There exists $n_0\in\mathbb{N}$ such that for every $n\ge n_0$ there exist distributions $p,q$ on $\{0,1\}^n$ with
\begin{equation}
    \mathrm{TVD}(p,q)=1,\qquad \forall\,\alpha\in\{0,1\}^n\setminus\{\mathbf{0}\}:\ |\phi_p(\alpha)| \le 2^{-n/4},\ \ |\phi_q(\alpha)| \le 2^{-n/4},
\end{equation}
where $\mathbf{0}$ is the bitstring with zeroes in all positions and $\phi_p, \phi_q$ are the characteristic functions of the corresponding distributions.
\end{lemma}

\begin{proof}
Let $N=2^n$ and set $m:=\lfloor N/3\rfloor$. Sample $A\subseteq\{0,1\}^n$ uniformly among all subsets of size $m$, and let $B:= \{0, 1\}^n \setminus A$. Define $p:=\mathcal{U}_A$ and $q:=\mathcal{U}_B$, where $\mathcal{U}_S$ denotes the uniform distribution on $S$.

The characteristic function of $p$ can be written as
\begin{equation}
\phi_p(\alpha) = \mathbb{E}_{x \sim \mathcal{U}_A}[(-1)^{\alpha \cdot x}] = \frac{1}{m} \sum_{x \in A} (-1)^{\alpha \cdot x},
\end{equation}
where $\alpha\cdot x=\sum_{i=1}^n \alpha_i x_i \pmod 2$. Using Hoeffding's inequality for sampling without replacement, and since $\mathbb{E}_{x \sim \mathcal{U}_{\{0, 1\}^n}}[(-1)^{\alpha \cdot x}] = 0$ for $\alpha\neq\mathbf{0}$ while each summand lies in $[-1,1]$, we have for $\alpha \neq \mathbf{0}$:
\begin{equation}
    \Pr\left(|\phi_p(\alpha)|\ge \varepsilon\right)\ \le\ 2\exp \Big(-\frac{\varepsilon^2 m}{2}\Big).
\end{equation}
A union bound over $2^n - 1 < 2^n$ non-zero $\alpha$ gives:
\begin{equation}
    \Pr\left(\exists \alpha \neq \mathbf{0} : |\phi_p(\alpha)| \geq \varepsilon \right) \leq 2^{n+1} \exp\left(-\frac{\varepsilon^2 m}{2} \right).
\end{equation}
With $\varepsilon=2^{-n/4}$ and $m=\lfloor N/3\rfloor$,
\begin{equation}
    \frac{\varepsilon^2 m}{2} \geq \frac{2^{-n/2}}{2}\cdot\frac{2^n}{3} = \frac{2^{n/2}}{6},
\end{equation}
so
\begin{equation}
    \Pr\left(\exists\,\alpha\neq \mathbf{0}: |\phi_p(\alpha)|\geq 2^{-n/4}\right) \leq 2^{n+1}\exp\Big(-\frac{2^{n/2}}{6}\Big),
\end{equation}
which is $<2^{-n}$ for all sufficiently large $n$. Thus, with probability at least $1-2^{-n}$ (over the draw of $A$),
\begin{equation}
    |\phi_p(\alpha)|\le 2^{-n/4}\qquad\forall\,\alpha\neq \mathbf{0}.
\end{equation}

The characteristic function of $q$ can be expressed via that of $p$:
\begin{equation}
    0=\phi_{\mathcal{U}_{\{0,1\}^n}}(\alpha)=\frac{m}{N}\phi_p(\alpha)+\frac{N-m}{N}\phi_q(\alpha) \quad\Rightarrow\quad \phi_q(\alpha)=-\frac{m}{N-m}\,\phi_p(\alpha).
\end{equation}
Because $m\le N/3$, we have $\frac{m}{N-m}\le \frac{1}{2}$, hence 
\begin{equation}
    |\phi_q(\alpha)| \leq \frac{1}{2}|\phi_p(\alpha)| \leq 2^{-n/4-1} \leq 2^{-n/4}\qquad\forall\,\alpha\neq \mathbf{0}.
\end{equation}
Finally, $A\cap B=\emptyset$ implies $\mathrm{TVD}(p,q)=1$. This establishes the claim for all sufficiently large $n$.
\end{proof}

\section{Numerical experiments details}
\begin{figure*}[tb]
\centering
\includegraphics[width=\textwidth]{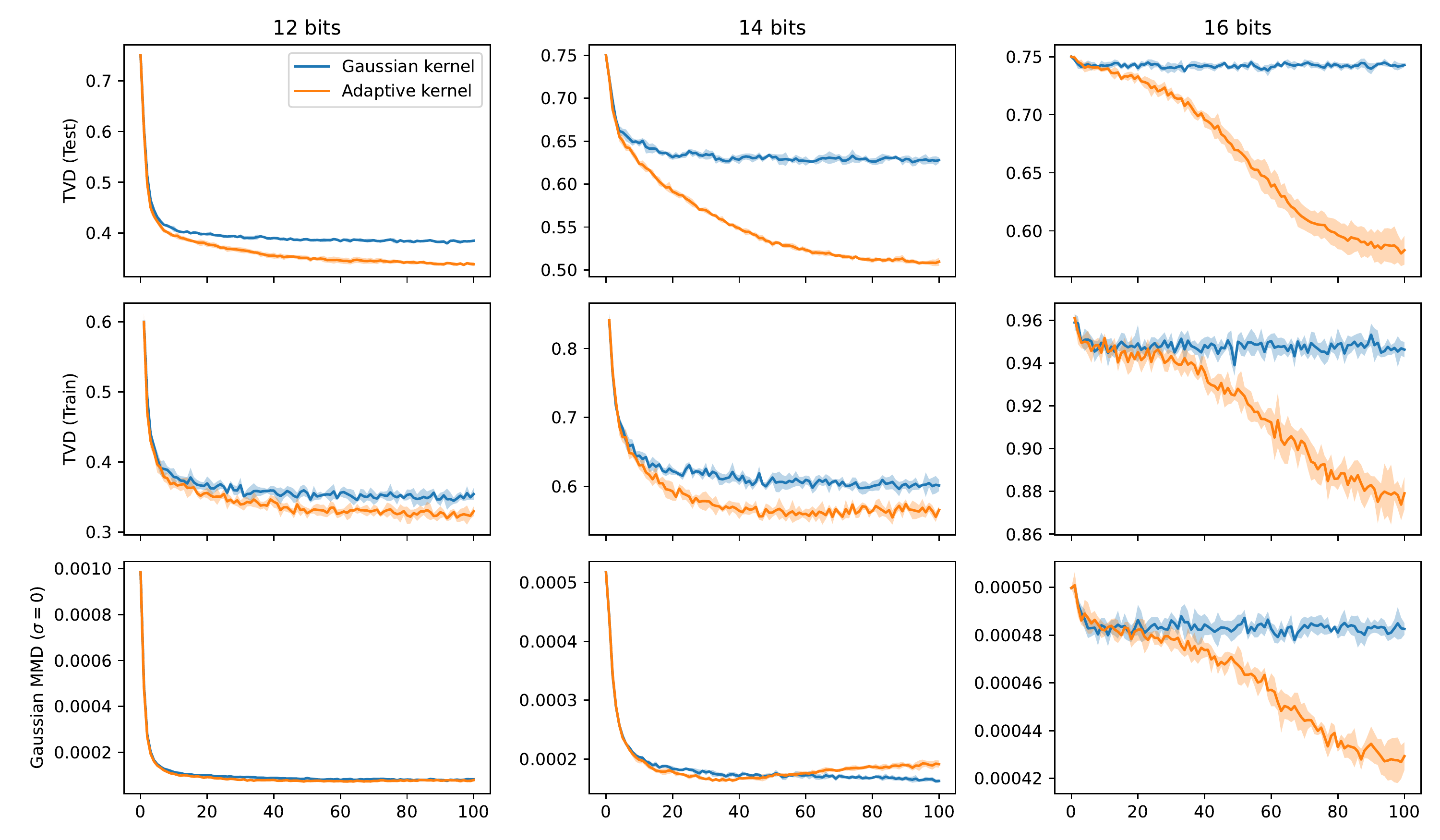}
\caption{Learning curves for (mean $\pm$ standard deviation over $5$ runs) on synthetic parity-check datasets of dimension $12$, $14$ and $16$. The first row shows the total variation distance computed between the true generated distribution and ground truth data distribution (same as Figure 2 in the main text). The second row shows the total variation distance computed between the empirical distribution of a batch of $1000$ generated data and the empirical distribution of the $2000$ training data. The third row shows the empirical MMD with respect to the $0$-bandwidth Gaussian kernel, based on $1000$ generated data and $2000$ training data.}
\label{fg: extented result}
\end{figure*}
Here we describe the detailed settings for the numerical experiments with the synthetic parity check datasets. The parity-check matrices behind the ground truth distributions are \begin{align}
H_{12}=&\begin{bmatrix} 111111110000 \\
000011111111 \end{bmatrix} \\ H_{14}=& \begin{bmatrix} 11111111100000 \\
00000111111111 \end{bmatrix} \\ H_{16}=&\begin{bmatrix} 1111111111000000 \\
0000001111111111 \end{bmatrix}
\end{align} respectively. To prepare the training and testing datasets, bit strings $x$ are uniformly randomly sampled from $\{0,1\}^n$, accepted into the dataset if $Hx=0$, and rejected if $Hx=1$. 

All experiment runs share the same hyperparameters: We use $2000$ training data. To compute the MMD losses during training, we use batches of $1000$ frequencies. For each frequency, we compute the quantum part of the MMD losses (Z expectation values) using $1000$ measurements, and the data part of the MMD losses using all training data. For updating the IQP generator, we use the Adam optimizer with learning rate $0.001$ and $\beta_1 = 0.9, \beta_2=0.999$. For updating the adaptive kernel, we use the Adam optimizer with learning rate $0.0001$ (decay by $10\%$ every $500$ iterations) and $\beta_1 = 0.9, \beta_2=0.999$. For the adaptive training runs, we update the generator and kernel at a rate of $1:1$. 

Now we show how FVSBN, the model we use to implement the parameterized spectral measure, can be initialized as the Gaussian spectral measure and the warm spectral measure. To recover a Gaussian spectral measure with bandwidth $\sigma$, we set $W=0$, and $\forall 1 \leq i \leq n$, $b_i= \log(p-\epsilon) - \log(1-p-\epsilon)$ where $p=(1-\exp(-1/(2\sigma^2)))/2$. Note that for the special case of $\sigma=0$, we simply take $b_i=0$. To prepare a warm spectral measure, $W$ and $b$ need to be set as some specific sparse matrices and vectors. As an example, for the $n=12$ case, we set $b_1 = \ln(2)$, $b_5=\frac{K}{2}$, $b_9=b_{10}=b_{11}=b_{12} = W_{2,1} = W_{3,1} = W_{4,1} = W_{6,5} = W_{7,5} = W_{8,5} = K$, $W_{9,1} = W_{10,1} = W_{11,1} = W_{12,1} = W_{9,5} = W_{10,5} = W_{11,5} = W_{12,5} = -K$ and $W_{5,1} = -\frac{K}{2}$, where we choose $K=10$. ($K$ can be chosen between $0$ and $\infty$. A larger $K$ makes the spectral warmer, i.e., more concentrated on the three frequencies.) 

We show the detailed experimental results in~\cref{fg: extented result}. We observe that for the 12- and 14-bit cases, both Gaussian ($\sigma=0$) and Adaptive ($\sigma=0$) methods converge to a significantly lower total variation distance (TVD) on the training set, while yielding substantially higher TVD on the test set, indicating overfitting. For the 16-bit case, the Gaussian ($\sigma=0$) method shows a flat learning curve since (1) the $2000$ training points are quite sparse in ${0,1}^{16}$ and (2) its spectral measure is uniform, which fails to reflect the difference between the generator and the training set. The Adaptive ($\sigma=0$) method is initialized from the Gaussian and overfits to the training data as it converges to a much lower training TVD. As for the MMD loss, we show the Gaussian MMD ($\sigma=0$) values for these two methods in the bottom row of~\cref{fg: extented result}. Both Gaussian ($\sigma=0$) and Adaptive ($\sigma=0$) converge to near-zero values.

\end{document}